\documentclass[10pt,tightenlines,eqsecnum,floats,aps,amsmath,
amssymb,nofootinbib,prd,showpacs]{revtex4-1}
\usepackage{amsmath,amsthm,latexsym,amssymb,amsfonts}
\usepackage{enumerate}
\usepackage{graphicx}
\usepackage{calc}
\usepackage[usenames,dvipsnames]{color}
\usepackage{colordvi}
\usepackage{vmargin}
\newcommand{\Z}{{\mathbb Z}}
\newcommand{\R}{{\mathbb R}}
\newcommand{\C}{{\mathbb C}}

\newcommand{\N}{{\mathbb N}}
\newcommand{\omicron}{{\cal I}}
\newcommand{\Hil}{{\mathcal H}}
\newcommand{\clpn}[1]{\left\lbrack #1\right\rbrack}
\newtheorem{lm}{Lemma}
\newtheorem{thm}[lm]{Theorem}
\newtheorem{df}[lm]{Definition}

\newcommand{\gora}[7] %(1,2)pkt-(3,4) (5,6), 7-podpis
{
  \put(0,0){\qbezier(#1,#2)(#1,#4)(#3,#4)}
\put(0,0){\qbezier(#3,#4)(#5,#4)(#5,#6)}
\put(#3,#4){#7}
}
\newcommand{\pudla} % (1,2) gorny punkt
{
\newsavebox{\pudlo}
\savebox{\pudlo}(2,1)[bl]{
\put(0,0){\line(1,0){1.5}}
\put(0,0){\line(0,1){0.8}}
\put(1.5,0){\line(0,1){0.8}}
\put(0,0.8){\line(1,0){1.5}}
\put(0.5,0.2){$P$}}
}

\begin{document}
%%%%%%%%%%%%%%%%%%%%%%%%%%%%%%%%%%

\title{The EPRL intertwiners and corrected partition function} 

\author{Wojciech Kami\'nski, Marcin Kisielowski, Jerzy Lewandowski}

\affiliation{Instytut Fizyki Teoretycznej, Uniwersytet Warszawski,
ul. Ho{\.z}a 69, 00-681 Warszawa (Warsaw), Polska (Poland)\\
%${}^2$Institute for Gravitation and the Cosmos \&
%Physics Department, Penn State, University Park, PA 16802, U.S.A.
}

\begin{abstract} \noindent{\bf Abstract\ }
Do the SU(2) intertwiners parametrize the space of the EPRL solutions
to the simplicity constraint? What is a complete form of the partition
function written in terms of this parametrization? We prove that the
EPRL map is injective for n-valent vertex in case when 
it is a map from $SO(3)$ into $SO(3)\times SO(3)$ representations. We find, however, that the EPRL
map is not isometric. In the consequence, a partition function can be defined
either using the EPRL intertwiners Hilbert product or the SU(2) intertwiners
Hilbert product. We use the EPRL one and derive a new, complete formula for the
partition function. Next, we view it in terms of the SU(2) intertwiners. The
result, however, goes beyond the SU(2) spin-foam models
framework and the original EPRL proposal.
\end{abstract}

\pacs{%{04.60.Kz}
%\Red{(quantum gravity:minisuperspaces)},
  %{04.60.Pp}
  %\Red{(loop quantum gravity/spin-foams)},
  %{98.80.Qc}
  %\Red{(quantum cosmology)}
  }
\maketitle

\def\be{\begin{equation}}
\def\ee{\end{equation}}
\def\ba{\begin{eqnarray}}
\def\ea{\end{eqnarray}}
\def\lp{{\ell}_{\rm Pl}}
\def\g{\gamma}
\section{Introduction} The main technical ingredient of the spin-foam models of
4-dimensional gravity is so called quantum simplicity constraint. Imposing suitably defined constraint on the domain of the (discrete) path integral turns the  SU(2)$\times$SU(2) 
(or SL(2,C)) BF theory  into the spin-foam model of the Euclidean (respectively, Lorentzian) gravity \cite{EPRL}. The formulation of the simplicity  constraint believed to be correct, or at least fitting gravity the best among the known approaches \cite{graviton}, is the one derived by Engle, Pereira, Rovelli, Livine (EPRL) \cite{EPRL} (and independently derived by Freidel and Krasnov \cite{FK}).   
The solutions to the EPRL simplicity constraint are EPRL SU(2)$\times$SU(2) intertwiners. They are defined by the EPRL transformation, which
maps each SU(2) intertwiner into a EPRL solution of the simplicity constraint.
An attempt is made in the literature \cite{EPRL} to parametrize the space 
of the EPRL solutions by the SU(2) intertwiners. The vertex amplitude and the partition function of the EPRL model seem
to written  in  in terms of that parametrization.    
The questions we raise and answer in this paper are: 
\begin{itemize}
\item Is the EPRL map injective,  doesn't it kill any SU(2) intertwiner?
\item Is the EPRL map isometric, does it preserve the scalar product between
      the SU(2) intertwiners? 
\item If not, what is a form of a partition function derived from
the SO(4) intertwiner Hilbert product  written
directly in terms of the  SU(2) intertwiners, the preimages of the
EPRL map?
\end{itemize}
We prove that the
EPRL map is injective for n-valent vertex in case when 
it is a map from $SO(3)$ into $SO(3)\times SO(3)$ representations. The full proof  of injectivity of EPRL map (without additional assumptions from this 
paper) in the case
$|\gamma| >1$ has already been provided in \cite{SFLQG}. In those cases, there are as many SU(2)$\times$SU(2) EPRL intertwiners as there are the SU(2) intertwiners. Owing to this result the  SU(2) intertwiners indeed can be used
to parametrize the space of the EPRL SU(2)$\times$SU(2)intertwiners.   
 However, we find the EPRL map  is not isometric. In consequence, there are two
inequivalent  definitions of the partition function. One possibility is to use a
basis in the EPRL intertwiners space orthonormal with respect to the SO(4)
representations. 
 And this is what we do in this paper. A second possibility, is to use the basis
obtained as the image of an orthonormal basis of the 
 SU(2) intertwiners under the EPRL map. The partition function derived  in \cite{EPRL} corresponds to the second choice, whereas the first one is ignored therein. The goal of this part of our paper  is pointing out the first possibility and deriving
the corresponding partition function.  After the derivation, we compare our
partition function with that of EPRL on a possibly simple example. We
conjecture, that the difference converges to zero for large spins.
 
To make the paper intelligible we start presentation of the new results with the
derivation of the partition function in Section \ref{problem}. The final formula
for our proposal for the partition function for the EPRL model is presented in
Section \ref{solution}. The lack of the isometricity   
of the EPRL map is illustrated on specific examples in Section \ref{example}. Finally,  the proof of the injectivity of the EPRL map takes all the  Section \ref{proof}.   

This work is written in terms of the EPRL framework \cite{EPRL} combined with our  
previous paper \cite{SFLQG} on the EPRL model.                
             
\section{Our proposal for a partition function of the EPRL model}          
\subsubsection{Partition functions for the spin-foam models of 4-gravity: definition}
Consider an oriented 2-complex $\kappa$ whose faces (2-cells) are labeled by $\rho$ with the irreducible representations of  $G=$SU(2)$\times$SU(2), 
$$ \kappa^{(2)}\ni f\ \mapsto \rho(f), $$
and denote by  ${\cal H}(f)$ the corresponding Hilbert space. 
For every edge (1-cell) $e$ we have the set/set of incoming/outgoing faces, that is the faces which contain $e$ and whose orientation agrees/disagrees with the orientation of $e$. We use them to define the Hilbert space
\be {\cal H}(e)\ =\ \bigotimes_{f\ {\rm incoming}}{\cal H}(f)\,\otimes\,\bigotimes_{f'\ {\rm outgoing}}{\cal H}(f')^*. \ee 
The extra data we use, is a subspace 
\be {\cal H}(e)^{{\rm SIMPLE}}\ \subset\ {\cal H}(e)\label{simple}\ee
defined by some constraints called the quantum simplicity constraints.
In this paper, starting from section below, we will be considering the subspace  proposed by Engle-Pereira-Rovelli-Livine. For the time being ${\cal H}(e)^{{\rm SIMPLE}}$ is any subspace  of the space of invariants of the 
representation $\bigotimes_{f\ {\rm incoming}}\rho(f)\,\otimes\,\bigotimes_{f'\ {\rm outgoing}}\rho(f')^*$,  
\be {\cal H}(e)^{\rm SIMPLE} \subset\ {\rm Inv}_{SU(2)\times SU(2)}\left(\bigotimes_{f\ {\rm incoming}}{\cal H}(f)\,\otimes\,\bigotimes_{f'\ {\rm outgoing}}{\cal H}(f')^*\right).\ee   
(The subspace ${\cal H}_{e}^{\rm SIMPLE}$ may be trivial  for generic representations  $\rho(f)$ and $\rho(f')$. Typically the simplicity constraints constrain also the
representations themselves.)
To every edge we assign the operator  of the orthogonal projection onto ${\cal H}(e)^{\rm SIMPLE}$,
\be P_e^{\rm SIMPLE}\ :\ \bigotimes_{f\ {\rm incoming}}{\cal H}(f)\,\otimes\,\bigotimes_{f'\ {\rm outgoing}}{\cal H}(f')^* \rightarrow\  \bigotimes_{f\ {\rm incoming}}{\cal H}(f)\,\otimes\,\bigotimes_{f'\ {\rm outgoing}}{\cal H}(f')^*\ee
Our index notation is as follows (we drop `SIMPLE' for simplicity)
\be (P_ev)^{A...}{}_{B...}\ =\ {P_e}^{A...B'...}_{A'...B...}v^{A'...}{}_{B'...}\ee
where the upper/lower indices of any vector $v\in \bigotimes_{f\ {\rm incoming}}{\cal H}(f)\,\otimes\,\bigotimes_{f'\ {\rm outgoing}}{\cal H}(f')^*$ correspond to
incoming/outgoing faces. In the operator $P_e$, for each face containing $e$, there are two indices, an upper and a lower one corresponding to the Hilbert space ${\cal H}(f)$. If $f$ is incoming (outgoing), then we assign the corresponding lower/upper index of $P_e$ to the beginning/end (end/beginning) of the edge. 
That rule is illustrated on Fig. \ref{contraction}.   
%%%%%%%%%%%%%%%%%%%%%%%%%%%%%%%%%%%%%%%%%%%%%%%%%%%%%%%%%%%%%%%%%%%%%%%%%%%%%%%%%%%%
\begin{figure}[ht!]
  \centering
    \includegraphics[width=\textwidth]{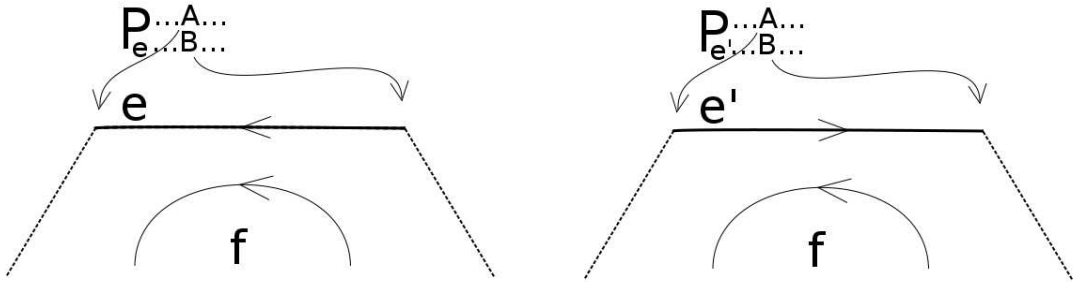}
\caption{According to this rule, given an edge $e$ ( $e'$ ) contained in incoming (outgoing) face $f$,
the indices of $P_e$ ( $P_{e'}$ ) corresponding to ${\cal H}(f)$ are assigned  to the
beginning and, respectively, to the end of the edge.  The oriented arc only marks the orientation of the polygonal face $f$.}
{\label{contraction}}
\end{figure}
%%%%%%%%%%%%%%%%%%%%%%%%%%%%%%%%%%%%%%%%%%%%%%%%%%%%%%%%%%%%%%%%%%%%%%%%%%%%%%%%%%%%%%
Now, for every pair of edges $e$ and $e'$, which belong to a same face $f$, and share a vertex $v$, there is defined the natural contraction at $v$ of the corresponding vertex of $P_e$ with the corresponding vertex of  $P_{e'}$. The contraction
defines the following trace,
\be \bigotimes_{e\in\kappa^{(1)}}P_e^{\rm SIMPLE}\ \mapsto\ {\rm
Tr}\left(\bigotimes_{e\in\kappa^{(1)}}P_e^{\rm SIMPLE}\right).\label{tr} \ee 
 Define  partition  function  $Z(\kappa,\rho)$ to be the following
number:
\be Z(\kappa,\rho)\ :=\ \prod_{f\in \kappa^{(2)}}d(f)\,
{\rm Tr}\left(\bigotimes_{e\in\kappa^{(1)}}P_e^{\rm SIMPLE}\right)\,{\cal
A}({\rm boundary}),\ \ \ \ \ d(f)\ :=\ {\rm dim}{\cal H}(f)\label{ZP}\ee
 where ${\cal A}({\rm boundary})$ is a factor that depends only
on the boundary of $(\kappa,\rho)$, and we derive it elsewhere.
\subsubsection{Partition functions for the spin-foam models of 4-gravity: the amplitude form}
The partition function is usually rewritten in the spin-foam amplitude form \cite{RR, Baezintro, perez}. 
For that purpose one needs an orthonormal basis in each  Hilbert space ${\cal H}(e)^{\rm SIMPLE}$; denote its elements  by  
$\iota_{e,\alpha} \in {\cal H}(e)^{\rm SIMPLE}$,  $\alpha = 1,2,...,n(e)$. Then 
\be P_e^{\rm SIMPLE}\ =\ \sum_{\alpha=1}^{n_e}\iota_{e,\alpha}\otimes\iota^\dagger_{e,\alpha}\label{unity} \ee 
where by $`\dagger'$, for every Hilbert space ${\cal H}$ we denote the duality map
$$  {\cal H}\ni v\ \mapsto\ v^\dagger\in {\cal H}^*$$
defined by the Hilbert scalar product. In the Dirac notation 
$$\iota_{e,\alpha}\ =\ |e,\alpha\rangle, \ \ {\rm and}\ \ \iota^\dagger_{e,\alpha}\ =\ \langle e,\alpha|. $$
Substituting  the right hand side of (\ref{unity}) for each $P_e^{\rm SIMPLE}$ in (\ref{ZP}),  one writes the partition function in terms of the vertex amplitudes
 in the following way:
\begin{itemize}
\item For each edge of $\kappa$ choose an element  of the corresponding orthonormal basis; denote this assignment by
\be \iota\ :\ e\ \rightarrow\ \iota_{e,\alpha_e}\label{iota}\ee
\item  At each vertex $v\in\kappa^{(0)}$:
\begin{itemize}
\item take $\iota_{e_1,\alpha_{e_1}},...,\iota_{e_m,\alpha_{e_m}}$ where $e_1,...,e_m$ are the incoming edges
\item take $\iota_{e'_1,\alpha_{e'_1}}^\dagger,...,\iota_{e'_{m'},\alpha_{e'_{m'}}}^\dagger$  where $e'_1,...,e'_{m'}$ are the outgoing edges
\item define the vertex amplitude 
\be A_v(\iota)\ :=\ {\rm Tr}\left(\iota_{e_1,\alpha_1}\otimes...\otimes\iota_{e_m,\alpha_m}\,\otimes\,
\iota_{e'_1,\alpha'_1}^\dagger\otimes...\otimes\iota_{e'_{m'},\alpha'_{m'}}^\dagger\right) \ee
where `Tr' stands for the contraction (\ref{tr}) and can be defined by  the
evaluation of the spin-networks corresponding to the vertices (see \cite{SFLQG}).
%%%%%%%%%%%%%%%%%%%%%%%%%%%%%%%%%%%%%%%%%%%%%%%%%%%%%%%%%%%%%%%%%%%%%%%%%%%%%%%%%%%%
%\begin{figure}[ht!]
 % \centering
  %  \includegraphics[width=\textwidth]{vertex.eps}
%\caption{?}
%{\label{vertex}}
%\end{figure}
%%%%%%%%%%%%%%%%%%%%%%%%%%%%%%%%%%%%%%%%%%%%%%%%%%%%%%%%%%%%%%%%%%%%%%%%%%%%%%%%%%%%%%

\end{itemize}
\item to each face $f$ assign the face amplitude $d(f).$
\end{itemize}
With this data, with the vertex amplitudes and face amplitudes, the partition function takes the famous form
\be Z(\kappa,\rho)\ =\  \prod_{f\in
\kappa^{(2)}}d(f)\,\sum_{\iota}\prod_{v\in\kappa^{(0)}}A_v(\iota)\,{\cal A}({\rm
boundary}),\label{ZA}\ee   
 The result is independent of the choice of the orthonormal basis 
 of each ${\cal H}_e^{\rm SIMPLE}$.
 
\subsubsection{The EPRL map}  
Now we turn to the EPRL intertwiners. For every edge $e\in \kappa^{(1)}$ 
\be {\cal H}(e)^{\rm SIMPLE}\ =\  {\cal H}(e)^{\rm EPRL},\ \ \ \ 
{P}(e)^{\rm SIMPLE}\ =\  {P}(e)^{\rm EPRL}. \ee
The definition of ${\cal H}(e)^{\rm EPRL}$ uses a fixed number $\gamma\in\mathbb{R}$ called Barbero-Immirzi parameter. The Hilbert space ${\cal H}(e)^{\rm EPRL}$ can be non-empty only if the 2-complex $\kappa$ is labeled by EPRL representations. A representation $\rho=(\rho_{j^-},\rho_{j^+})$ of SU(2)$\times$SU(2), where $j^\pm\ \in\ \frac{1}{2}\mathbb{N}$ define the SU(2) representations in the usual way,  is an EPRL representation, provided there is $k\in \frac{1}{2}\mathbb{N}$ such that 
\be j^\pm\ =\ \frac{|1\pm \gamma|}{2}k.\label{k}\ee
Therefore, we will be considering here labellings of the faces of the 2-complex $\kappa$ with
EPRL representations
\be f\ \mapsto\ \rho(f)\ =\ (\rho_{j^-(f)},\rho_{j^+(f)}), \ee
each of which is defined in the Hilbert space  
$${\cal H}(f)\ =\ {\cal H}_{j^-(f)}\otimes {\cal H}_{j^+(f)}.$$
Each labeling defines also a labeling with SU(2) representations given by (\ref{k}),
\be f\ \mapsto \rho_{k(f)}, \ee
defined in the Hilbert space ${\cal H}_{k(f)}$. 
Given an edge $e$ and the corresponding Hilbert space 
\be {\cal H}(e)\ =\ \bigotimes_{f\ {\rm incoming}}{\cal H}_{j^-(f)}\otimes {\cal H}_{j^+(f)}  \,\otimes\,\bigotimes_{f'\ {\rm outgoing}}{\cal H}_{j^-(f')}^*\otimes {\cal H}_{j^+(f')}^*\label{He},\ee 
% $\left(\,(\rho_{j^-_1},\rho_{j^+_1}),...,(\rho_{j^-_M},\rho_{j^+_M})\ \right)$ of EPRL representations, and the  sequence $\left(\,{\cal H}_{j^-_1}\otimes{\cal H}_{j^+_1},...,{\cal H}_{j^-_M}\otimes{\cal H}_{j^+_M}\, \right)$ 
%of the corresponding Hilbert spaces, 
the natural isometric embeddings 
\be C\ : \ {\cal H}_k\ \rightarrow {\cal H}_{j^-}\otimes {\cal H}_{j^+}\ee
and the orthogonal projection operator
\be P:\ {\cal H}(e)\ \rightarrow  {\cal H}(e) \ee
onto the subspace ${\rm Inv}_{SU(2)\times SU(2)}\left( {\cal H}(e) \right)$
defines the natural map, the EPRL map:    
\begin{align*}
\iota^{\rm EPRL}:\ {\rm Inv}_{SU(2)}\big( \bigotimes_{f\ {\rm incoming}}{\cal H}_{k(f)}\,&\otimes\,\bigotimes_{f'\ {\rm outgoing}}{\cal H}_{k(f')}^* \big)\\ &\rightarrow\  \bigotimes_{f\ {\rm incoming}}{\cal H}_{j^-(f)}\otimes {\cal H}_{j^+(f)}  \,\otimes\,\bigotimes_{f'\ {\rm outgoing}}{\cal H}_{j^-(f')}^*\otimes {\cal H}_{j^+(f')}^*
\end{align*}
Its image is promoted to  the Hilbert space (\ref{simple}),
\be  {\cal H}_e^{\rm EPRL}\ :=\ {\iota^{\rm EPRL}}\left(
{\rm Inv}_{SU(2)}\left( \bigotimes_{f\ {\rm incoming}}{\cal H}_{k(f)}\,\otimes\,\bigotimes_{f'\ {\rm outgoing}}{\cal H}_{k(f')}^*
\right)\right).\ee

\subsubsection{The problem with the EPRL intertwines}\label{problem}
All the EPRL intertwiners can be constructed from the SU(2)
intertwiners by using the EPRL map. The point
is, that  one has to be more careful while doing that. 
First,  one has to make sure that the map ${\iota^{\rm EPRL}}$ is injective. If not, then 
Hilbert space of the SU(2)$\times$SU(2) EPRL intertwiners is smaller then the corresponding space of the SU(2) intertwiners and we should know how big it is. For $\gamma\ge 1$, the injectivity was proved in \cite{SFLQG}. In the next section we present a proof of the injectivity for $|\gamma| <1$.
Secondly, one should check whether or not the map ${\iota^{\rm EPRL}}$ is isometric. Given an orthonormal basis $\omicron_{e,1},...,\omicron_{e,n_e}$ of  the Hilbert space ${\rm Inv}_{SU(2)}\left( \bigotimes_{f\ {\rm incoming}}{\cal H}_{k(f)}\,\otimes\,\bigotimes_{f'\ {\rm outgoing}}{\cal H}_{k(f')}^* \right)$,
we have a corresponding basis \\${\iota^{\rm
EPRL}}(\omicron_{e,1}),...,{\iota^{\rm EPRL}}(\omicron_{e,n_e})$ of the
corresponding Hilbert space 
${\cal H}^{\rm EPRL}_e$. The question is, whether or not the latter basis is also orthonormal. We show in Section \ref{example}, this is not the case.  The direct procedure would be  to orthonormalize the  basis. We propose, however, a simpler solution.
\subsubsection{A solution}\label{solution}
 An intelligent way, is to go back
to the formula (\ref{ZP}) for the partition function and repeat the step leading to
(\ref{ZA}) with each projection $P_e^{\rm SIMPLE}=P_e^{\rm EPRL}$ written in terms of the corresponding basis  ${\iota^{\rm EPRL}}(\omicron_{e,1}),...,{\iota^{\rm EPRL}}(\omicron_{e,n_e})$. The suitable formula reads
\be P_e^{\rm EPRL}\ =\ \sum_{a,b=1}^{n_e}  h_e{}^{a\bar{b}}{\iota^{\rm EPRL}}(\omicron_{e,a})\otimes{\iota^{\rm EPRL}}(\omicron_{e,b})^\dagger\label{peprl} \ee
where $h_e{}^{a\bar{b}}$, $a,b=1,...,n_e$ define the inverse matrix to the matrix
\be h_{e,\bar{b}a}\ :=\ ({\iota^{\rm EPRL}}(\omicron_{e,b})|{\iota^{\rm EPRL}}(\omicron_{e,a}))\label{h}\ee
given by the Hilbert product $(\cdot|\cdot)$  in the Hilbert space (\ref{He}).
(In the Dirac notation, $P_e^{\rm EPRL}\ =\   h_e{}^{a\bar{b}}|e,a\rangle\langle e,b|$.)    

Now, we are in a position to write the resulting spin-foam amplitude formula for the partition function. It is assigned to a fixed  2-complex $\kappa$ and a fixed labeling of the faces by the EPRL representations 
\be f\ \mapsto\ \rho(f)\ =\ (\rho_{j^-_f},\rho_{j^+_f}). \ee
The labeling is accompanied by the corresponding labeling with the SU(2) i
\be f\ \mapsto\ \rho_{k_f}, \ee
according to (\ref{k}). For every edge $e\in\kappa^{(1)}$, in addition to the 
Hilbert space ${\cal H}_e^{\rm EPRL}\subset {\rm Inv}_{SU(2)\times SU(2)}\left({\cal H}_e\right)$, we also have its preimage, 
the Hilbert space  
%\be {\cal H}_e^{\rm EPRL}\ =\ \bigotimes_{f\ {\rm incoming}}{\cal %H}_{\rho_f}\,\otimes\,\bigotimes_{f'\ {\rm outgoing}}{\cal H}_{\rho_{f'}}^*,
% \ee
\be {\rm Inv}\left( \bigotimes_{f\ {\rm incoming}}{{\cal H}_{(k_f)}}\,\otimes\,\bigotimes_{f'\ {\rm outgoing}}{{\cal H}}_{k_{f'}}^*\right). \ee
Therein, we fix  an orthonormal basis
\be \omicron_{e,a}, \ \ \ \ a\ =\ 1,2,...,n_e.\ee 
To define the partition function we proceed as follows:
\begin{itemize}
\item 
assign to every edge of $\kappa$  a pair of elements of the basis, 
\be \omicron\omicron:  e \ \mapsto\ (\omicron_{e,a_e},\omicron_{e,b_e}^\dagger), \ee
more specifically, $\omicron_{e,a_e}$ is assigned to the end point and $\omicron_{e,b_e}^\dagger$ to the beginning point of $e$, and we denote the assignment by the double symbol $\omicron\omicron$;
\item define for every edge an edge amplitude to be 
$$h_e(\omicron\omicron)\ :=\  h^{b_e\bar{a}_e}$$
\item to every vertex $v$ of $\kappa$ assign the vertex amplitude with the trace
defined by Fig.1, (\ref{tr}) and (\ref{peprl})  
$$A_v(\omicron\omicron)\ :=\ {\rm Tr}\left(\bigotimes_{e\ {\rm incoming}}{\iota^{\rm EPRL}}(\omicron_{e,a_e})\,\otimes\,\bigotimes_{e'\ {\rm outgoing}}{\iota^{\rm EPRL}}(\omicron_{e',b_{e'}})^\dagger\right)  $$
\item to every  face $f$ assign the amplitude $d_f$
\end{itemize}
Finally, the spin-foam amplitude formula for the partition function reads
\be Z(\kappa,\rho)\ =\ \prod_f d_f \sum_{\omicron \omicron} \prod_eh_e(\omicron\omicron)\prod_v
 A_v(\omicron\omicron){\cal A}({\rm boundary}).\label{pffinal}\ee

The matrix (\ref{h}) can be written in terms of the EPRL fusion coefficients, 
\be {\iota^{\rm EPRL}}(\omicron_{e,a})\ =:\ {f_{e}}_a^{c^-c^+}\iota_{e,c^-}\otimes\iota_{e,c^+}, \ee
defined by the decomposition into an orthonormal basis 
$\iota_{e,c^-}\otimes\iota_{e,c^+}$ in\\
${\rm Inv}_{SU(2)\times SU(2)}\left( \bigotimes_{f\ {\rm incoming}}{\cal
H}_{j^-(f)}\otimes {\cal H}_{j^+(f)}  \,\otimes\,\bigotimes_{f'\ {\rm
outgoing}}{\cal H}_{j^-(f')}^*\otimes {\cal H}_{j^+(f')}^* \right).$\\
Then, we have
\be h_{e,\bar{a}b}\ =\ \sum_{c^+,c^-}\overline{{f_{e}}_a^{c^-c^+}}{f_{e}}_b^{c^-c^+}.\ee

 We are in a position now, to compare our partition function with
that of \cite{EPRL}. The partition function of \cite{EPRL} is given by 
replacing the matrix $h_{e,\bar{a}b}$ in (\ref{pffinal}) with $\delta_{ab}$.
  The example below gives quantitative idea about the difference between the two
possible definitions of  partition function
 for the EPRL model.

%%%%%%%%%%%%%%%%%%%%%%%%%%%%%%%%%%
\subsubsection{Example of the edge amplitude $h^{\bar{b}a}$  showing
that the EPRL map is not isometry}\label{example}
We will show in this section, that the EPRL map is not isometric.
We  calculate the edge amplitude defined in the previous section
in a simple example, and see that its matrix is not proportional 
to the identity matrix, or even not diagonal.    
Consider an edge at which exactly  four faces meet. Assume the orientation of each face is opposite to the orientation of the edge. We have an intertwiner $\omicron \in {\rm Inv}_{SU(2)}\left( {\cal H}_{k_1}\otimes{\cal H}_{k_2}\otimes {\cal H}_{k_3}\otimes {\cal H}_{k_4} \right)$ assigned to the end point of this edge. We choose a basis $|k_i m_i \rangle$ (the eigenvector of the third component of angular momentum operator with eigenvalue $m_i$) in each space ${\cal H}_{k_i}$, $i\in \{1,\ldots,4\}$. We choose a real basis of the space $ {\rm Inv}_{SU(2)}\left( {\cal H}_{k_1}\otimes{\cal H}_{k_2}\otimes {\cal H}_{k_3}\otimes {\cal H}_{k_4} \right)$ in the following form \cite{Fusion_Asymptotics}:
\[
(\omicron_a)_{k_1 m_1 k_2 m_2 k_3 m_3 k_4 m_4}=\sqrt{2a+1}\sum_{m=-a}^{a}\sum_{m'=-a}^{a}(-1)^{a+m}\left (
\begin{array}{ccc}
k_1&k_2&a\\
m_1&m_2&m
\end{array}
\right )
\delta_{m,-m'}
\left (
\begin{array}{ccc}
k_3&k_4&a\\
m_3&m_4&m'
\end{array}
\right )
\]
where $\left (
\begin{array}{ccc}
j_1&j_2&j_3\\
m_1&m_2&m_3
\end{array}
\right )$ is the Wigner 3j-Symbol, $\delta_{m,m'}$ is the Kronecker Delta.\\
Let $\iota_{a^+}\otimes \iota_{a^-}$ be the basis of the space $ {\rm Inv}\left( {\cal H}_{j_1^+}\otimes\ldots\otimes {\cal H}_{j_4^+}\right) \otimes {\rm Inv}\left( {\cal H}_{j_1^-}\otimes\ldots \otimes {\cal H}_{j_4^-} \right)$. The intertwiner $\iota^{\rm EPRL}(\omicron_a)$ expressed in this basis takes the following form
(we skip in this section the subscript $e$ indicating the dependence on edge):
\[
{\iota^{\rm EPRL}}(\omicron_a)=f_{a}^{a^+ a^-} \iota_{a^+} \otimes \iota_{a^-}
\]
where $f_{a}^{a^+ a^-}$ are real and are known as fusion coefficients \cite{EPRL}.
The tensor $h_{\bar{a}b}$ could be expressed in terms of them:
\[
h_{\bar{a}b}=({\iota^{\rm EPRL}}(\omicron_a)|{\iota^{\rm EPRL}}(\omicron_b) )=\sum_{a^+ a^-} f_{a}^{a^+ a^-} f_{b}^{a^+ a^-}
\]
As an example we give the result of the calculation of the $h_{a\bar{b}}$ matrix for $\gamma=\frac{1}{2},j_1=2,j_2=4,j_3=4,j_1=2;a,b\in \{2,\ldots,6 \}$): 
\[
\left(
\begin{array}{ccccc}
 \frac{53723 \underline{}}{175616} & -\frac{2265 \sqrt{\frac{5}{7}} \underline{}}{50176} & \frac{5093 \sqrt{5} \underline{}}{1053696} & -\frac{3 \sqrt{55} \underline{}}{25088} & 0 \\
 -\frac{2265 \sqrt{\frac{5}{7}} \underline{}}{50176} & \frac{117853 \underline{}}{501760} & -\frac{12805 \underline{}}{301056 \sqrt{7}} & \frac{45 \sqrt{\frac{11}{7}} \underline{}}{7168} & -\frac{3 \sqrt{\frac{13}{7}} \underline{}}{8960} \\
 \frac{5093 \sqrt{5} \underline{}}{1053696} & -\frac{12805 \underline{}}{301056 \sqrt{7}} & \frac{741949 \underline{}}{3512320} & -\frac{781 \sqrt{11} \underline{}}{752640} & \frac{5 \sqrt{13} \underline{}}{5376} \\
 -\frac{3 \sqrt{55} \underline{}}{25088} & \frac{45 \sqrt{\frac{11}{7}} \underline{}}{7168} & -\frac{781 \sqrt{11} \underline{}}{752640} & \frac{583 \underline{}}{2560} & 0 \\
 0 & -\frac{3 \sqrt{\frac{13}{7}} \underline{}}{8960} & \frac{5 \sqrt{13} \underline{}}{5376} & 0 & \frac{13 \underline{}}{40}
\end{array}
\right)
\]
We used the analytic expression for the fusion coefficient presented in \cite{Fusion_Asymptotics}. Clearly this matrix is nondiagonal. It shows that the EPRL map is not isometric. The edge amplitude $h_{e}{}^{b\bar{a}}$ is given by the inverse matrix:
\[\left(
\begin{array}{ccccc}
 \frac{46976713}{14064543 \underline{}} & \frac{31728718 \sqrt{\frac{7}{5}}}{70322715 \underline{}} & -\frac{75194882}{257849955 \sqrt{5} \underline{}} & -\frac{3865813}{70322715 \sqrt{55} \underline{}} & \frac{13066606}{773549865 \sqrt{65} \underline{}} \\
 \frac{31728718 \sqrt{\frac{7}{5}}}{70322715 \underline{}} & \frac{7682388364}{1758067875 \underline{}} & \frac{67078172 \sqrt{7}}{586022625 \underline{}} & -\frac{318127222 \sqrt{\frac{7}{11}}}{1758067875 \underline{}} & \frac{7212044 \sqrt{\frac{7}{13}}}{1758067875 \underline{}} \\
 -\frac{75194882}{257849955 \sqrt{5} \underline{}} & \frac{67078172 \sqrt{7}}{586022625 \underline{}} & \frac{112636131412}{23636245875 \underline{}} & \frac{1305090458}{6446248875 \sqrt{11} \underline{}} & -\frac{12462294236}{70908737625 \sqrt{13} \underline{}} \\
 -\frac{3865813}{70322715 \sqrt{55} \underline{}} & -\frac{318127222 \sqrt{\frac{7}{11}}}{1758067875 \underline{}} & \frac{1305090458}{6446248875 \sqrt{11} \underline{}} & \frac{85031744497}{19338746625 \underline{}} & -\frac{192524374}{19338746625 \sqrt{143} \underline{}} \\
 \frac{13066606}{773549865 \sqrt{65} \underline{}} & \frac{7212044 \sqrt{\frac{7}{13}}}{1758067875 \underline{}} & -\frac{12462294236}{70908737625 \sqrt{13} \underline{}} & -\frac{192524374}{19338746625 \sqrt{143} \underline{}} & \frac{8510451083428}{2765440767375 \underline{}}
\end{array}
\right)
\]
\\

\section{Injectivity of the map $\omicron\mapsto\iota^{\rm EPRL}(\omicron)$} 
\label{Sec_inject}
\setlength{\unitlength}{1cm}

\pudla
This part of the paper is devoted to the injectivity of EPRL intertwiner. More explicitly, we will prove the result stated in \ref{inj-EPRL}.
\subsection{Statement of the result}
We assume that $\gamma\in\R$ and $|\gamma|<1$. Suppose that
\begin{align*}
 &(k_1,\ldots, k_n)\in{\frac{1}{2}} \N\\
&\forall_i\ \ j^\pm_i=\frac{1\pm\gamma}{2} k_i\in {\frac{1}{2}} \N.
\end{align*}
We consider the EPRL map
\begin{align}
\label{Def}
&\iota^{\rm EPRL}\colon  {\rm Inv}\left(\Hil_{k_1}\otimes\cdots\otimes\Hil_{k_n}\right)\rightarrow {\rm Inv}\left(\Hil_{j^+_1}\otimes\cdots\otimes\Hil_{j^+_n}\right)\otimes {\rm Inv}\left(\Hil_{j^-_1}\otimes\cdots\otimes\Hil_{j^-_n}\right)\\
&\iota^{\rm EPRL}(\omicron)_{j^+_1 A_1\ldots j^+_n A_n j^-_1 B_1\ldots j^-_n B_n}=
\omicron_{k_1 C_1\ldots k_n C_n} C^{k_1 C_1}_{j^+_1 D_1 j^-_1 E_1}\cdots C^{k_n C_n}_{j^+_n D_n j^-_n E_n} P^{j^+_1 D_1\ldots j^+_n D_n}_{j^+_1 A_1\ldots j^+_n A_n} P^{j^-_1 E_1\ldots j^-_n E_n}_{j^-_1 B_1\ldots j^-_n B_n}\nonumber
\end{align}
with $P$ standing for the orthogonal projections onto the subspaces of 
the SU(2) invariants of the Hilbert spaces $\Hil_{j^+_1}\otimes\cdots\otimes\Hil_{j^+_n}$, and respectively, 
$\Hil_{j^-_1}\otimes\cdots\otimes\Hil_{j^-_n}$.

Now we can state our result.
\begin{thm}
\label{inj-EPRL}
For any sequence $(k_1,\ldots, k_n)\in \N$ such that $(j_1^+,\ldots, j_n^+)\in \N$, the map
\[
 \iota^{\rm EPRL}\colon  {\rm Inv}\left(\Hil_{k_1}\otimes\cdots\otimes\Hil_{k_n}\right)\rightarrow {\rm Inv}\left(\Hil_{j^+_1}\otimes\cdots\otimes\Hil_{j^+_n}\right)\otimes {\rm Inv}\left(\Hil_{j^-_1}\otimes\cdots\otimes\Hil_{j^-_n}\right)
\]
is injective.
\end{thm}
 In other words we consider here the case, when 
  the EPRL map carries SO(3) representations into 
  SO(3) x SO(3) representations.
%In the case $|\gamma|=1$ this is result is obvious.

\subsection{Proof of the theorem}\label{proof}
In order to make the proof transparent, we divide it into subsections. In subsection \ref{inj-aux} some auxiliary definitions are introduced. We state also an inductive hypothesis, that will be proved in subsection \ref{inj-proof}. The injectivity of EPRL map follows from that result. The main technical tool of the proof is placed in subsection \ref{inj-map}, where the map \ref{mapG} is defined. 
\subsubsection{Auxiliary definitions}
\label{inj-aux}
Let us introduce some notations
\begin{df}
For $x\in\R$ we define 
\begin{itemize}
\item $\clpn{x}_+$ as the only integer number in the interval $\left(x-\frac{1}{2},x+\frac{1}{2}\right]$
\item $\clpn{x}_-$ as the only integer number in the interval $\left[x-\frac{1}{2},x+\frac{1}{2}\right)$
\end{itemize}
\end{df}
and
\begin{df}
A sequence of half natural numbers $(k_1,\ldots, k_n)$ satisfies {\it triangle inequality} if
\[
 \forall_i\ \ k_i\le \sum_{j\not=i} k_j.
\]
\end{df}

%We assume that $|\gamma|<1$. 
One can define map $\iota$ under condition  
\begin{quote}
{\bf Con $n$:} Sequences of natural numbers $(k_1,\ldots, k_n)$ and $(j^\pm_1,\ldots, j^\pm_n)$ are such that
\begin{itemize}
 \item $(k_1,\ldots, k_n)$ satisfies triangle inequality,
\item $j^+_i+j^-_i=k_i$ for $i=1,\ldots, n$,
\item $j^\pm_i=\frac{1\pm\gamma}{2} k_i$ for $i\not=1$ and
\[
 j^+_1=\clpn{\frac{1+\gamma}{2} k_1}_+,\ j^-_1=\clpn{\frac{1-\gamma}{2} k_1}_-
\]
or
\[
 j^+_1=\clpn{\frac{1+\gamma}{2} k_1}_-,\ j^-_1=\clpn{\frac{1-\gamma}{2} k_1}_+.
\]
\end{itemize}
\end{quote}
Let us define
\begin{align*}
&\iota_{k_1\ldots k_n}\colon  {\rm Inv}\left(\Hil_{k_1}\otimes\cdots\otimes\Hil_{k_n}\right)\rightarrow {\rm Inv}\left(\Hil_{j^+_1}\otimes\cdots\otimes\Hil_{j^+_n}\right)\otimes {\rm Inv}\left(\Hil_{j^-_1}\otimes\cdots\otimes\Hil_{j^-_n}\right)\\
&\iota_{\ldots}(\omicron)_{j^+_1 A_1\ldots j^+_n A_n j^-_1 B_1\ldots j^-_n B_n}=
\omicron_{k_1 C_1\ldots k_n C_n} C^{k_1 C_1}_{j^+_1 D_1 j^-_1 E_1}\cdots C^{k_n C_n}_{j^+_n D_n j^-_n E_n} P^{j^+_1 D_1\ldots j^+_n D_n}_{j^+_1 A_1\ldots j^+_n A_n} P^{j^-_1 E_1\ldots j^-_n E_n}_{j^-_1 B_1\ldots j^-_n B_n}
\end{align*}
with $P$ standing for projections onto invariant subspaces. We will use the letter $\iota_{k_1\ldots k_n}$ for all sequences $(k_1,\ldots,k_n)$, $(j^\pm_1,\ldots,j^\pm_n)$ if it do not cause any misunderstanding.

We will base our prove on the following inductive hypothesis:
\begin{quote}
{\bf Hyp $n$:} Suppose that $(k_1,\ldots, k_n)$ and $(j^\pm_1,\ldots, j^\pm_n)$ satisfy condition {\bf Con $n$} and that $\omicron\in 
{\rm Inv}\left(\Hil_{k_1}\otimes\cdots\otimes\Hil_{k_n}\right)$. Then, there 
exists 
$$\phi\in {\rm Inv}\left(\Hil_{j^+_1}\otimes\cdots\otimes\Hil_{j^+_n}\right)\otimes {\rm Inv}\left(\Hil_{j^-_1}\otimes\cdots\otimes\Hil_{j^-_n}\right)$$ such that
$\langle \iota_{k_1\ldots k_n}(\omicron),\phi\rangle\not=0$.
\end{quote}
This in fact proves the injectivity.

\subsubsection{Useful inequalities}

Both $\clpn{x}_\pm$ are increasing functions and satisfy ($x,y\in\R$, $j\in \N$)
\begin{itemize}
\item[a)] $\clpn{x+j}=\clpn{x}+j$ and $\clpn{j}_\pm=j$,
%\item if $x\ge n$ then $\clpn{x}_\pm\ge n$,
\item[b)] if $x>y$ then $\clpn{x}_-\ge\clpn{y}_+$ and if $x\ge y$ then $\clpn{x}_+\ge\clpn{y}_-$
\item[c)] if $x+y\in \Z$ then $\clpn{x}_++\clpn{y}_-=x+y$
\item[d)] if $x+y\ge j$ then $\clpn{x}_++\clpn{y}_-\ge j$
\end{itemize}
In order to prove the last point, we notice that $\clpn{x}_+> x-\frac{1}{2}$ and $\clpn{y}_-\ge y-\frac{1}{2}$ so $\clpn{x}_++\clpn{y}_->x+y-1\ge j-1$ but as $j$ is an integer number $\clpn{x}_++\clpn{y}_-\ge j$.

\begin{lm}\label{triangle}
Suppose that $(k,l,j)$ satisfies triangle inequality and that $\frac{1+\gamma}{2}k\in \N$ then 
\begin{itemize}
\item both triples $\left(\frac{1+\gamma}{2}k, \clpn{\frac{1+\gamma}{2}l}_\pm, \clpn{\frac{1+\gamma}{2}j}_\pm\right)$ satisfy triangle inequalities if $k+l=j$ or $k+j=l$
\item both triples $\left(\frac{1+\gamma}{2}k, \clpn{\frac{1+\gamma}{2}l}_\pm, \clpn{\frac{1+\gamma}{2}j}_\mp\right)$ satisfy triangle inequalities if $k+l>j$ and $k+j>l$.
\end{itemize}
\end{lm}

\begin{proof}
In the first case suppose that $k+l=j$ holds, then $\frac{1+\gamma}{2}k+\clpn{\frac{1+\gamma}{2}l}_\pm= \clpn{\frac{1+\gamma}{2}j}_\pm$ that proves triangle inequality.

In the second case, we restrict our attention to $\left(\frac{1+\gamma}{2}k, \clpn{\frac{1+\gamma}{2}l}_+, \clpn{\frac{1+\gamma}{2}j}_-\right)$.
We have

\begin{itemize}
\item $\frac{1+\gamma}{2}k + \clpn{\frac{1+\gamma}{2}j}_-\ge \clpn{\frac{1+\gamma}{2}l}_+$ because $\frac{1+\gamma}{2}k+{\frac{1+\gamma}{2}l}> {\frac{1+\gamma}{2}j}$,
\item $\frac{1+\gamma}{2}k + \clpn{\frac{1+\gamma}{2}l}_+\ge \clpn{\frac{1+\gamma}{2}j}_-$
because $\frac{1+\gamma}{2}k+{\frac{1+\gamma}{2}j}\ge {\frac{1+\gamma}{2}l}$,
\item $\frac{1+\gamma}{2}k\le \clpn{\frac{1+\gamma}{2}l}_++ \clpn{\frac{1+\gamma}{2}j}_-$ from the property d) listed above.
\end{itemize}
The case of $\left(\frac{1+\gamma}{2}k, \clpn{\frac{1+\gamma}{2}l}_-, \clpn{\frac{1+\gamma}{2}j}_+\right)$ is analogous.
\end{proof}

\begin{lm}
\label{triangle2}
Suppose that $(k_1,\ldots,k_n)$ satisfies triangle inequality and that $\frac{1+\gamma}{2}k_i\in \N$ for $i=2,\ldots,n$, then $\left(\clpn{\frac{1+\gamma}{2}k_1}_\pm, \frac{1+\gamma}{2}k_2,\ldots,\frac{1+\gamma}{2}k_n\right)$ also satisfy triangle inequalities.
\end{lm}

\begin{proof}
That follows from the monotonicity of functions $\clpn{x}_\pm$ and the fact that in the inequality
\[
 \frac{1+\gamma}{2}k_i\le \sum_{j\not=i} \frac{1+\gamma}{2}k_j
\]
all terms but one are integer.
\end{proof}

\begin{lm}
Suppose that $(k_1,\ldots, k_n)\in \N$ satisfies triangle inequality then
\[
 {\rm Inv}\left(\Hil_{k_1}\otimes\cdots\otimes\Hil_{k_n}\right)
\]
is nontrivial.
\end{lm}

\begin{proof}
We will find an $k_\alpha\in \N$ such that both $(k_\alpha,k_1,k_2)$ and $(k_\alpha,k_3,\ldots, k_n)$ satisfy triangle inequalities.  We then have $k_\alpha+k_1+k_2\in \N$ and $(k_\alpha,k_3,\ldots, k_n)\in \N$. By induction there would be 
\[
 0\not=\phi\in {\rm Inv} \left(\Hil_{k_\alpha}\otimes \Hil_{k_3}\otimes\cdots\Hil_{k_n}\right),
\]
and then
\[
 0\not=C^{k_\alpha A}_{k_1 A_1 k_2 A_2} \phi_{k_\alpha A k_3 A_3\ldots k_n A_n}
\]
proves nontriviality. Now we extract conditions on $k_\alpha$ from triangle  inequalities
(we assume for simplicity that $k_1\ge k_2$)
\begin{align*}
&k_1+k_2\ge k_\alpha \ge k_1-k_2\\
&\sum_{i\ge 3} k_i\ge k_\alpha\ge k_i-\sum_{j\not=i,\ j\ge 3} k_j,\ i\ge 3
\end{align*}
For the existence of such $k_\alpha$ we need only to show that
\begin{align*}
 &k_1+k_2\ge k_i-\sum_{j\not=i,\ j\ge 3} k_j,\ i\ge 3\\
&\sum_{i\ge 3} k_i\ge k_1-k_2
\end{align*}
but these are exactly conditions for $(k_1,\ldots, k_n)$ to satisfy triangle inequality.
\end{proof}

\subsubsection{Important maps}
\label{inj-map}
Every $\omicron\in {\rm Inv}\left(\Hil_{k_1}\otimes\cdots\otimes\Hil_{k_n}\right)$ may be uniquely written as
\[
 \omicron_{k_1 A_1 k_2 A_2\ldots k_n A_n}
=\sum_{k_{\alpha}}C^{k_\alpha A_\alpha}_{k_1 A_1 k_2 A_2}
\omicron^{k_\alpha}_{k_\alpha A_\alpha k_3 A_3\ldots k_n A_n},
\]
where $\omicron^{k_\alpha}\in {\rm Inv}\left(\Hil_{k_\alpha}\otimes\Hil_{k_3}\otimes\cdots\otimes\Hil_{k_n}\right)$. Summation is taken over such $k_\alpha$ that $(k_\alpha,k_1,k_2)$ and $(k_\alpha,k_3,\ldots,k_n)$ satisfy triangle inequality, $k_\alpha+k_1+k_2\in \N$.

This gives us decomposition of ${\rm Inv}\left(\Hil_{k_1}\otimes\cdots\otimes\Hil_{k_n}\right)$ into orthogonal subspaces
\[
 \oplus_\alpha \Hil_\alpha,
\]
where each $\Hil_\alpha$ is isomorphic to ${\rm Inv}\left(\Hil_{k_\alpha}\otimes\Hil_{k_3}\otimes\cdots\otimes\Hil_{k_n}\right)$.
Let us define maps which assign these partial isometries 
\[
Q_{k_\alpha}\colon  {\rm Inv}\left(\Hil_{k_1}\otimes\cdots\otimes\Hil_{k_n}\right)\rightarrow {\rm Inv}\left(\Hil_{k_\alpha}\otimes\Hil_{k_3}\otimes\cdots\otimes\Hil_{k_n}\right),\ \
Q_{k_\alpha}\omicron=\omicron^{k_\alpha}
\]
Adjoints to them are embeddings $Q^*_{k_\alpha}$.
\[
 Q^*_{k_\alpha}\left(\omicron^{k_\alpha}\right)_{k_1 A_1 k_2 A_2\ldots k_n A_n}
=C^{k_\alpha A_\alpha}_{k_1 A_1 k_2 A_2}
\omicron^{k_\alpha}_{k_\alpha A_\alpha k_3 A_3\ldots k_n A_n}.
\]
These maps are also well defined in a case that $\alpha$ does not occur in the decomposition $\oplus\Hil_\alpha$ but $(k_\alpha,k_1,k_2)$ satisfies triangle inequalities. Then the space ${\rm Inv}\left(\Hil_{k_\alpha}\otimes\Hil_{k_3}\otimes\cdots\otimes\Hil_{k_n}\right)$ is trivial and the maps $Q_{k_\alpha}$ and $Q^*_{k_\alpha}$ too.

Let us fix $(k_1,\ldots, k_n)$ and $(j^\pm_1,\ldots, j^\pm_n)$ satisfying triangle inequalities and such that $j^+_i+j^-_i=k_i$.
\begin{lm}\label{mapG}
Suppose $k_\alpha,j^\pm_\alpha$ are such that $j^+_\alpha+j^-_\alpha=k_\alpha$ and $(k_\alpha,k_1,k_2)$ and $(j^\pm_\alpha,j^\pm_1,j^\pm_2)$ satisfy triangle inequalities, $k_\alpha+k_1+k_2\in \N$, $j^\pm_\alpha+j^\pm_1+j^\pm_2\in \N$. Then there exists an operator 
\begin{align*}
G_{k_\alpha j^+_\alpha j^-_\alpha}\ \colon 
{\rm Inv}\big(\Hil_{j^+_\alpha}\otimes\Hil_{j^+_3}\otimes\cdots&\otimes\Hil_{j^+_n}\big)\otimes {\rm Inv}\big(\Hil_{j^-_\alpha}\otimes\Hil_{j^-_3}\cdots\otimes\Hil_{j^-_n}\big)
\rightarrow \\ &\rightarrow
{\rm Inv}\big(\Hil_{j^+_1}\otimes\cdots\otimes\Hil_{j^+_n}\big)\otimes
{\rm Inv}\big(\Hil_{j^-_1}\otimes\cdots\otimes\Hil_{j^-_n}\big) 
\end{align*}
such that for all $\omicron\in {\rm Inv}\left(\Hil_{k_\beta}\otimes\Hil_{k_3}\otimes\cdots\otimes\Hil_{k_n}\right)$ and $\phi\in {\rm Inv}\left(\Hil_{j^+_\alpha}\otimes\Hil_{j^+_3}\otimes\cdots\otimes\Hil_{j^+_n}\right)\otimes {\rm Inv}\left(\Hil_{j^-_\alpha}\otimes\Hil_{j^-_3}\cdots\otimes\Hil_{j^-_n}\right)$
\[
 \langle\iota_{k_1\ldots k_n}Q^*_{k_\beta}\omicron,\ G_{k_\alpha j^+_\alpha j^-_\alpha}\phi\rangle=
\left\{\begin{array}{ll}
\langle\iota_{k_\alpha k_3\dots k_n}\omicron,\ \phi\rangle, & k_\beta=k_\alpha\\ 
        0, & k_\beta>k_\alpha\\
*, & k_\beta<k_\alpha
       \end{array}\right.
\]
\end{lm}

\begin{proof}
We define $G_{k_\alpha j^\pm_\alpha}$ as
\[
 G_{k_\alpha j^\pm_\alpha}\left(\phi\right)_{j^+_1 A_1\ldots j^+_n A_n j^-_1 B_1\ldots j^-_n \ldots B_n}=
\beta C^{j^+_\alpha A_\alpha}_{j^+_1 A_1 j^+_2 A_2}C^{j^-_\alpha B_\alpha}_{j^-_1 A_1 j^-_2 A_2}
\phi_{j^+_\alpha A_\alpha j^+_3 A_3\ldots j^-_\alpha B_\alpha j^-_3 B_3\ldots},
\]
with $\beta$ nonzero constant to be defined later.
Let us compute
\[
 \langle\iota_{k_1\ldots k_n}Q^*_{k_\beta}\omicron, G_{k_\alpha j^\pm_\alpha}\phi\rangle
\]
In the definition of $\iota$ one can skip projection because both $\phi$ and $G_{k_\alpha j^\pm_\alpha}\phi$ are invariants. Let us write explicitly $\langle\iota_{k_1\ldots k_n}Q^*_{k_\beta}\omicron,\ G_{k_\alpha j^\pm_\alpha}\phi\rangle$. We have
\begin{align*}
\omicron_{ k_\beta A \ldots k_n A_n} C^{k_\beta A}_{k_1 A_1 k_2 A_2}& C^{k_1 A_1}_{j^+_1 B_1 j^-_1 C_1}\cdots C^{k_n A_n}_{j^+_n B_n j^-_n C_n}\ \ \beta C^{j^+_1 B_1 j^+_2 B_2}_{j^+_\alpha B} C^{j^-_1 C_1 j^-_2 C_2}_{j^-_\alpha C}\phi^{j^+_\alpha B j^+_3 B_3\ldots j^-_\alpha C j^-_3 C_3\ldots}=\\
&=\beta C^{k_\beta A}_{k_1 A_1 k_2 A_2} C^{k_1 A_1}_{j^+_1 B_1 j^-_1 C_1} C^{k_2 A_2}_{j^+_2 B_2 j^-_2 C_2}C^{j^+_1 B_1 j^+_2 B_2}_{j^+_\alpha B} C^{j^-_1 C_1 j^-_2 C_2}_{j^-_\alpha C}\\ &\omicron_{k_\beta A \ldots k_n A_n}  C^{k_3 A_3}_{j^+_3 B_3 j^-_1 C_3}\cdots C^{k_n A_n}_{j^+_n B_n j^-_n C_n}\ \ \phi^{j^+_\alpha B j^+_3 B_3\ldots j^-_\alpha C j^-_3 C_3\ldots}.
\end{align*}
We need only to show that
\[
\beta C^{k_\beta A}_{k_1 A_1 k_2 A_2} C^{k_1 A_1}_{j^+_1 B_1 j^-_1 C_1} C^{k_2 A_2}_{j^+_2 B_2 j^-_2 C_2}C^{j^+_1 B_1 j^+_2 B_2}_{j^+_\alpha B} C^{j^-_1 C_1 j^-_2 C_2}_{j^-_\alpha C}=
\left\{\begin{array}{ll}
C^{k_\beta A}_{j^+_\alpha B j^-_\alpha C}, & k_\beta=j^+_\alpha+j^-_\alpha\\
0 & k_\beta>j^+_\alpha+j^-_\alpha.
\end{array}\right.
\]
The second equality is obvious because there exists no intertwiner if $k_\beta>j^+_\alpha+j^-_\alpha$.
The first will be proved in the next subsection \ref{relint}
\footnote{Although it seems to be standard, we include it for a sake of completeness.}.
\end{proof}

\subsubsection{Relation among intertwiners}\label{relint}
We know that 
$C^{k_\beta A}_{k_1 A_1 k_2 A_2} C^{k_1 A_1}_{j^+_1 B_1 j^-_1 C_1} C^{k_2 A_2}_{j^+_2 B_2 j^-_2 C_2}C^{j^+_1 B_1 j^+_2 B_2}_{j^+_\alpha B} C^{j^-_1 C_1 j^-_2 C_2}_{j^-_\alpha C}$ is proportional to $C^{k_\beta A}_{j^+_\alpha B j^-_\alpha C}$. In order to prove that the factor of proportionality is nonzero we will show that
\[
C^{k_\alpha A}_{k_1 A_1 k_2 A_2} C^{k_1 A_1}_{j^+_1 B_1 j^-_1 C_1} C^{k_2 A_2}_{j^+_2 B_2 j^-_2 C_2}C^{j^+_1 B_1 j^+_2 B_2}_{j^+_\alpha B} C^{j^-_1 C_1 j^-_2 C_2}_{j^-_\alpha C}\ \ C_{k_\alpha A}^{j^+_\alpha B j^-_\alpha C}\not=0
\]
and that would be $\beta^{-1}$.
In fact it is enough to show that the intertwiner
\[
C^{k_1 A_1}_{j^+_1 B_1 j^-_1 C_1} C^{k_2 A_2}_{j^+_2 B_2 j^-_2 C_2}C^{j^+_1 B_1 j^+_2 B_2}_{j^+_\alpha B} C^{j^-_1 C_1 j^-_2 C_2}_{j^-_\alpha C}C_{k_\alpha A}^{j^+_\alpha B j^-_\alpha C}
\]
is not equal zero or equivalently the same for
\[
C^{k_1 A_1}_{j^+_1 B_1 j^-_1 C_1} C^{j^+_2 B_2 j^-_2 C_2}_{k_2 A_2}C^{j^+_1 B_1 }_{j^+_2 B_2j^+_\alpha B} C^{j^-_1 C_1 }_{j^-_2 C_2j^-_\alpha C}C_{k_\alpha A}^{j^+_\alpha B j^-_\alpha C} 
\]
We only sketch the proof. First of all, we remind some facts about intertwiners and diagrammatic notation.

Let $P^k$ stands for projection onto symmetric subspace in $\Hil_{1/2}^{\otimes 2k}$ equivalent to $\Hil_k$ ($k$ is a half natural number).
\[
 P^k\colon \Hil_{1/2}^{\otimes 2k}\rightarrow \Hil_{1/2}^{\otimes 2k}.
\]
In this subsection we regard $\Hil_k$ as this subspace of $\Hil_{1/2}^{\otimes 2k}$. Let us also denote the canonical map
$\epsilon\colon \C \mapsto \Hil_{1/2}\otimes \Hil_{1/2}$.
%%%%%%%%%%%%%

\begin{quote}
The intertwiner $C^{k_1k_2k_3}\colon \C\rightarrow \Hil_{k_1}\otimes \Hil_{k_2}\otimes\Hil_{k_3}$ is {\it proportional} to $P^{k_1}\otimes P^{k_2}\otimes P^{k_3}\epsilon^{\otimes 2k_1+2k_2+2k_3}$.
\end{quote}
In the diagrammatic language this can be depicted as on figure \ref{cint_0}. We skip the index $k$ in $P^k$ on the diagrams for notations' brevity. The line with symbol $k$ denotes $\Hil_{1/2}^{\otimes 2k}$.
\begin{figure}[ht!]
  \centering
\begin{picture}(4.5,3.8)
\put(0,0.5){\usebox{\pudlo}}
\put(3,0.5){\usebox{\pudlo}}
\put(1.5,2.5){\usebox{\pudlo}}
\gora{0.5}{1.3}{1}{1.8}{2}{2.5}{$k_{12}$}
\gora{4}{1.3}{3.5}{1.8}{2.5}{2.5}{$k_{13}$}
\gora{1}{1.3}{2.2}{1.6}{3.5}{1.3}{$k_{23}$}
\gora{0.7}{0}{0.7}{0.1}{0.7}{0.5}{$k_2$}
\gora{3.7}{0}{3.7}{0.1}{3.7}{0.5}{$k_3$}
\gora{2.3}{3.3}{2.3}{3.4}{2.3}{3.8}{$k_1$}
\end{picture}
\caption{An intertwiner {\it proportional} to $C^{k_1}_{k_2k_3}$, $k_{12}=k_1+k_2-k_3$ and etc.}
{\label{cint_0}}
\end{figure}
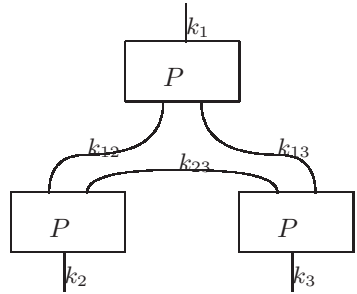

We have to notice important properties,that in diagrammatic language is shown on figures  \ref{cint_1} and \ref{cint_2}.

\begin{figure}[ht!]
\begin{picture}(2.5,2)
\gora{0}{0}{0}{1}{0}{2}{$k_1$}
\gora{0.8}{0}{0.8}{1}{0.8}{2}{$k_2$}
\put(1.6,1){$=$}
\gora{2.2}{0}{2.2}{1}{2.2}{2}{$k_1+k_2$}
\end{picture}
\caption{An equality between $\Hil_{1/2}^{2k_1}\otimes\Hil_{1/2}^{2k_2}$ and $\Hil_{1/2}^{2k_1+2k_2}$.}
{\label{cint_1}}
\end{figure}
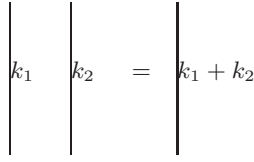
\begin{figure}[ht!]
\centering
\begin{picture}(4,4.5)
\gora{0.5}{0}{0.5}{0.3}{0.5}{1}{$k_1$}
\gora{0.5}{1.8}{0.5}{2.1}{0.5}{2.8}{$k_1$}
\gora{0.5}{3.6}{0.5}{4}{0.5}{4.5}{$k_1$}
\gora{0.8}{1.8}{0.8}{2.1}{0.8}{2.8}{$k_2$}
\gora{0.8}{0}{0.8}{0.3}{0.8}{1}{$k_2$}
\gora{0.8}{3.6}{0.8}{4}{0.8}{4.5}{$k_2$}
\gora{1.1}{0}{1.1}{0.3}{1.1}{1}{$k_3$}
\gora{1.1}{1.8}{1.8}{2.4}{2.2}{4.5}{$k_3$}
\put(0,1){\usebox{\pudlo}}
\put(0,2.8){\usebox{\pudlo}}
\put(2.2,2.2){$=$}
\gora{3.5}{1}{3.5}{1.3}{3.5}{2}{$k_1$}
\gora{3.8}{1}{3.8}{1.3}{3.8}{2}{$k_2$}
\gora{4.1}{1}{4.1}{1.3}{4.1}{2}{$k_3$}
\gora{3.5}{2.8}{3.5}{3.1}{3.5}{3.8}{$k_1$}
\gora{3.8}{2.8}{3.8}{3.1}{3.8}{3.8}{$k_2$}
\gora{4.1}{2.8}{4.1}{3.1}{4.1}{3.8}{$k_3$}
\put(3,2){\usebox{\pudlo}}
\end{picture}
\caption{An equality $P^{k_1+k_2+k_3}\circ P^{k_1+k_2}\otimes {\mathbb I}=P^{k_1+k_2+k_3}$.}
{\label{cint_2}}
\end{figure}
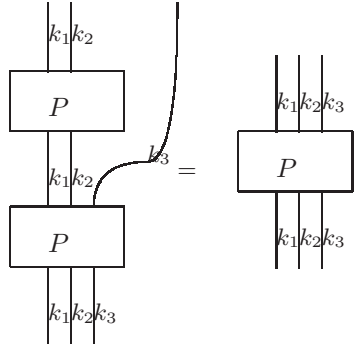

Our intertwiner can be written as shown on the figure \ref{cint_4}.

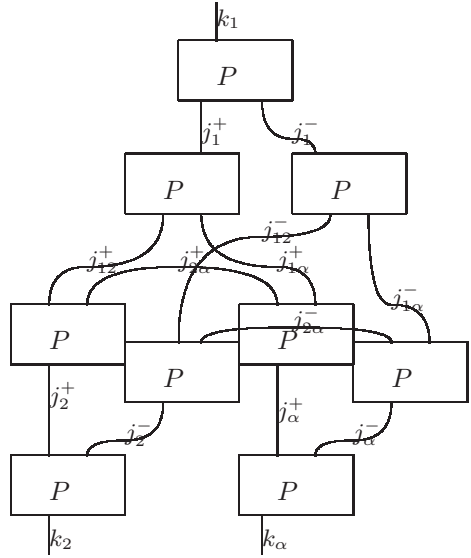
\begin{figure}[ht!]
  \centering
\begin{picture}(5.5,7.3)
\put(0,2.5){\usebox{\pudlo}}
\put(3,2.5){\usebox{\pudlo}}
\put(1.5,4.5){\usebox{\pudlo}}
\gora{0.5}{3.3}{1}{3.8}{2}{4.5}{$j^+_{12}$}
\gora{4}{3.3}{3.5}{3.8}{2.5}{4.5}{$j^+_{1\alpha}$}
\gora{1}{3.3}{2.2}{3.8}{3.5}{3.3}{$j^+_{2\alpha}$}
\put(1.5,2){\usebox{\pudlo}}
\put(4.5,2){\usebox{\pudlo}}
\put(3.7,4.5){\usebox{\pudlo}}
\gora{2.2}{2.8}{3.3}{4.2}{4.2}{4.5}{$j^-_{12}$}
\gora{5.5}{2.8}{5}{3.3}{4.7}{4.5}{$j^-_{1\alpha}$}
\gora{2.5}{2.8}{3.7}{3}{5}{2.8}{$j^-_{2\alpha}$}
\put(0,0.5){\usebox{\pudlo}}
\gora{0.5}{1.3}{0.5}{2}{0.5}{2.5}{$j^+_{2}$}
\gora{1}{1.3}{1.5}{1.5}{2}{2}{$j^-_{2}$}
\put(3,0.5){\usebox{\pudlo}}
\gora{3.5}{1.3}{3.5}{1.8}{3.5}{2.5}{$j^+_{\alpha}$}
\gora{4}{1.3}{4.5}{1.5}{5}{2}{$j^-_{\alpha}$}
\put(2.2,6){\usebox{\pudlo}}
\gora{2.5}{5.3}{2.5}{5.5}{2.5}{6}{$j^+_{1}$}
\gora{4}{5.3}{3.7}{5.5}{3.3}{6}{$j^-_{1}$}
\gora{0.5}{0}{0.5}{0.1}{0.5}{0.5}{$k_{2}$}
\gora{3.3}{0}{3.3}{0.1}{3.3}{0.5}{$k_{\alpha}$}
\gora{2.7}{6.8}{2.7}{7}{2.7}{7.3}{$k_1$}
\end{picture}
\caption{Intertwiner proportional to 
$C^{k_1 A_1}_{j^+_1 B_1 j^-_1 C_1} C^{j^+_2 B_2 j^-_2 C_2}_{k_2 A_2}C^{j^+_1 B_1 }_{j^+_2 B_2j^+_\alpha B} C^{j^-_1 C_1 }_{j^-_2 C_2j^-_\alpha C}C_{k_\alpha A}^{j^+_\alpha B j^-_\alpha C}$.}
{\label{cint_4}}
\end{figure}

Now using properties mentioned earlier 
we see that one can merge $j^+_{ij}$ with $j^-_{ij}$ into $k_{ij}$ and that intertwiner is equal to intertwiner shown on the figure \ref{cint_0} and is {\it nonzero}.

\subsubsection{Inductive steps}
\label{inj-proof}
Induction starts with $n=1$. In this case $k_1=0$ and so also $j^\pm_1=0$. The map
$\iota_{0}=C^0_{00}\colon \C\rightarrow \C\otimes\C$ is the identity.

Suppose now, that we have just proved {\bf Hyp $n-1$}.

In the decomposition of given $\omicron\in {\rm Inv}\left(\Hil_{k_1}\otimes\cdots\otimes\Hil_{k_n}\right)$ into subspaces $\Hil_\alpha$ we choose minimal $k_\alpha$ such that $Q_{k_\alpha}\omicron$ is nonzero. We know, by lemmas \ref{triangle} and \ref{triangle2} that for either
\[
(j^+_\alpha,j^-_\alpha)=\left(\clpn{\frac{1+\gamma}{2}k_\alpha}_+,\clpn{\frac{1-\gamma}{2}k_\alpha}_-\right)
\] 
or for
\[
(j^+_\alpha,j^-_\alpha)=\left(\clpn{\frac{1+\gamma}{2}k_\alpha}_-,\clpn{\frac{1-\gamma}{2}k_\alpha}_+\right)
\]
all necessary assumptions of lemma \ref{mapG} are satisfied. From the {\bf Hyp $n-1$} for the sequences $(k_\alpha,k_3,\ldots k_n)$, $(j^\pm_\alpha,j^\pm_3,\ldots j^\pm_n)$ we know that there exists \[
\phi\in {\rm Inv}\left(\Hil_{j^+_\alpha}\otimes\cdots\otimes\Hil_{j^+_n}\right)
\otimes {\rm Inv}\left(\Hil_{j^-_\alpha}\otimes\cdots\otimes\Hil_{j^-_n}\right)
\]
 such that
\[
 \left\langle\iota_{k_\alpha k_3\ldots k_n} Q_{k_\alpha}\omicron,\ \phi\right\rangle\not=0.
\]
We have
\[
 \left\langle\iota_{k_1\ldots k_n}\omicron,\ G_{k_\alpha j^\pm_\alpha}\phi\right\rangle=
\sum_{k_\beta\ge k_\alpha}\left\langle\iota_{k_1\ldots k_n} Q^*_{k_\beta}Q_{k_\beta}\omicron,\ G_{k_\alpha j^\pm_\alpha}\phi\right\rangle=\left\langle\iota_{k_\alpha k_3\ldots k_n} Q_{k_\alpha}\omicron,\ \phi\right\rangle\not=0.
\]
This finishes inductive step and the proof.
\section{Short Discussion}
We studied in this paper properties of the solutions to the EPRL simplicity constraints which were derived in \cite{EPRL}. We also pointed out two different
possibilities of defining the partition function out of them. Our definition
is (\ref{ZP}). It uses only the subspace of the SO(4) intertwiners which solve the
EPRL simplicity constraint. The comparison and contrast between our definition
and that of \cite{EPRL} is provided by (\ref{pffinal}) and the comments which
follow that equality. The difference follows from the fact proven in Section
\ref{example} above, that the EPRL map that is not isometric. The example
considered in that section gives also quantitative idea of the difference. The
question of which definition of the partition function is correct can not be
answered at this stage. Finally, we studied the ``size of the space of the EPRL
solutions''. We have shown that the EPRL map does not kill those SO(3) 
intertwiners, which are mapped into SO(3)$\times$SO(3) intertwiners. The proof is presented in detail in   Section \ref{Sec_inject}.
\bigskip

\noindent{\bf Acknowledgments} We would like to thank John Barrett, Jonathan Engle
for coming to Warsaw and delivering  lectures on the SFM. JL acknowledges the 
conversations with Laurent Freidel, and exchange of e-mails with John Baez, Carlo Rovelli, Roberto Pereira and Etera Livine. MK would like to thank Jonathan Engle from Albert Einstein Institute in Potsdam and Carlo Rovelli, Matteo Smerlak from Centre de Physique Theorique de Luminy for discussions and hospitality during his visits at their institutes. The work was partially supported by the Polish Ministerstwo Nauki i Szkolnictwa Wyzszego grants 182/NQGG/
2008/0,  2007-2010 research project N202 n081 32/1844, the National Science
Foundation (NSF) grant PHY-0456913, by the Foundation for Polish Science grant
Master and a Travel Grant from the QG research networking programme of the European Science Foundation.

\end{document}